\documentclass[11pt]{article}
\usepackage{geometry}                % See geometry.pdf to learn the layout options. There are lots.
\usepackage{graphicx}

\usepackage[utf8]{inputenc} % Any characters can be typed directly from the keyboard, eg éçñ
\usepackage{textcomp} % provide lots of new symbols

\usepackage{flafter}  % Don't place floats before their definition

\usepackage{amsmath,amssymb}  % Better maths support & more symbols
\usepackage{bm}  % Define \bm{} to use bold math fonts
\usepackage{amsthm}

\usepackage{memhfixc}  % remove conflict between the memoir class & hyperref
\usepackage{pdfsync}  % enable tex source and pdf output syncronicity

\newtheorem{definition}{Definition}[section]
\newtheorem{theorem}{Theorem}[section]
\newtheorem{proposition}{Proposition}[section]
\newtheorem{corollary}{Corollary}[section]
\newtheorem{remark}{Remark}[section]
\newtheorem{lemma}{Lemma}[section]
\numberwithin{equation}{section}

\title{From Kontsevich Graphs to Feynman graphs, a Viewpoint from the
       Star Products of Scalar Fields}
\author{Zhou Mai \footnote{address:Colleague of Mathematical Science, Nankai University, Weijin Road, Tianjin City, Republic China;
            email address: zhoumai@nankai.edu.cn}}

\begin{document}

\maketitle

\begin{abstract}
In the present paper a new approach to construct the star products
concerning the scalar fields is provided. Due to the structure
of the classical Hamiltonians, the star products will be developed
at three levels concerning functions, densities (fields) and functionals
respectively. 
The star product at level of functions is the starting point for our setting
which includes almost all information of the star products 
concerning the scalar fields and functionals. 
The point of the star product at level of functions
is that it only concerns the finite dimensional issue,
which is a Moyal-like star product on $\mathbb{C}^{\infty}(\mathbb{R}^{d})$ 
generated by a bi-vector field with abstract coefficients.  
Thus the Kontsevich graphs will play some roles naturally. 
Actually we prove that there is an ono-one correspondence
between a class of Kontsevich graphs and the Feynman graphs. Additionally the Wick theorem,
Wick power and the expectation of Wick-monomial are discussed in terms of the star product at level of functions.
Our construction can be considered as the generalisation of the star products in perturbative algebraic
quantum fields theory and twist product introduced in \cite{1},\cite{2}.
\end{abstract}

\tableofcontents

\section{Introduction}

The deformation quantisation of the fields is a infinite dimensional issue
essentially. Up to now, there are a lot of works about the deformation quantisation
in the infinite dimensional case (for example, see \cite{5},\cite{6},\cite{7},\cite{8},\cite{9},\cite{10},
where the list of references is not complete). In some sense the star products
in the infinite dimensional space were usually constructed as the copies of the classical
Moyal product, a typical point is that the 
partial derivatives in the classical Moyal bi-vector field are replaced by variational derivatives,
for example, Frechet derivative or others. If we focus on the deformation quantisation of 
the fields, for example, the case of scalar fields , we need to pay attention to the following
two facts. The first fact is the commutative relation (or Poisson bracket):

$$
\{\varphi(\textbf{x}),\varphi(\textbf{y})\}=K(\textbf{x},\textbf{y}),
$$
where $K(\textbf{x},\textbf{y})$ is some propagator and $\varphi(\textbf{x})$
is a scalar field in some physical theory. Above commutative relation suggests
the possibility of the Moyal-like product.
Another one is the variational calculation for a specific functional, for example

$$
F(\varphi)=\int f(\varphi(\textbf{x}_{1}),\cdots,
\varphi(\textbf{x}_{d}))d\textbf{x}_{1}\cdots d\textbf{x}_{d},
$$
the variation of $F(\varphi)$ can be calculated in terms of the partial derivatives of function
$\\$ $f(y_{1},\cdots,y_{d})$. This fact provides a possibility such that the infinite dimentional
calculations can be reduced to the finite dimentional situation.
Our setup is motivated by the facts mentioned above.

In the present article we will discuss a new approach to the deformation quantisation
of scalar fields in the covariant case, i.e. the Moyal-like star product.
An approach of the quantization of the scalar fields will
pass to Feynman amplitudes or Feynman diagrams with much possibility.
On the other hand, the Moyal-like star product will closely
connect with Kontsevich's gragh, where Poisson
bi-vector field will be reaplaced by a bi-vector field with
abstract coefficients. Thus, our construction
about the star product will result in the
connection between the Kontsevich graphs (see \cite{13}) and
the Feynman graphs. Other applications are discussed also. 
Now we make some explanation about function $f(y_{1},\cdots,y_{d})$
at level of terminology. In general $f(\varphi(\textbf{x}_{1}),\cdots,\varphi(\textbf{x}_{d}))$
(or more precisely $f(\varphi(\textbf{x}_{1}),\cdots,\varphi(\textbf{x}_{d}))dV$ where $dV$ is volume form)
is called density from the viewpoint of variational theory in classical fields theory. 
Here we need to distinguish between the $f(y_{1},\cdots,y_{d})$ and
$f(\varphi(\textbf{x}_{1}),\cdots,\varphi(\textbf{x}_{d}))$, so we call the function
$f(y_{1},\cdots,y_{d})$ the density function, or function for short.

Our approach to the star products is divided into three steps.
The first step is to construct  the star products at level of functions. 
As the second step, the star product of fields, or densities in the sense
mentioned above, can be constructed
from the first step simply. Finally, the star product at level 
of functionals can be costructed based on the second step.
We show our idea in the following table.
\begin{table}
\begin{center}
\begin{tabular}{c}
$f(x_{i})\star g(y_{j})$  \quad functions \\
   $\downarrow$   \\
$\begin{array}{c}
    f(\varphi)\star g(\varphi) \\
    = f(x_{i})\star g(y_{j})|_{x_{i}=\varphi(\cdot),y_{i}=\varphi(\cdot)}
  \end{array}$ 
 fields or densities  \\
$\downarrow$  \\
$\int f(\varphi)\star g(\varphi)$ \quad functionals
\end{tabular}
 \caption{default}
\end{center}
\label{defaulttable}
\end{table}
Our construction can be considered as a generalisation of covariant
deformation quantisation of the fields in perturbative algebraic quantum fields theory
(see \cite{7}, \cite{8}, \cite{9}) and twisted product introduced in  \cite{1}, \cite{2}.
Somehow the main outline of our construction is along the idea in our earlier work (see \cite{15}).

The basic starting point of our discussion is the construction of the star product
at level of functions. The deformation quantisation in the case of the fields
is the infinite dimensional issue basically. But in our approach, as a key point,
the construction of the star product at level of functions involves the finite dimensional issue only. 
Our discussion below will show that the star product of functions contains all algebraic and 
combinatorial information of deformation quantisation of the fields and functionals almost.
Actually, it will be showed that everything can be explicitly calculated based on
the calculations at level of functions almost.

Here the star product of functions is a  Moyal-like one in $\mathbb{R}^{d}$.   
The bi-vector field in the Moyal product is replaced by a bi-vector field
with abstract coefficients, this bi-vector field generates the star product
of the functions in our setting. With the help of the Moyal-like product 
in finite dimensional space our discussion
goes into the framework of Kontsevich naturally. We prove that for a special
class of the Kontsevich graphs, here we call that the graphs of Bernoulli type,
there is an one-one correspondence between the graphs of Bernoulli type
and the Feynman graphs. In this paper we consider only the Feynman
graphs without self-lines. It is well known that the Moyal product is the simplest
example in the theory of deformation quantisation on the Poisson manifolds,
thus the Kontsevich graphs involving the Moyal product should be the simplest
case. Our setting is completely parallel to the Moyal product from the viewpoint
of the Kontsevich graphs. Roughly speaking the set of the graphs of Bernoulli type
is generated by a special Bernoulli graph (see \cite{11}, \cite{12}) which may be the simplest,
 but non-trivial, graph even in Bernoulli graphs.
We will see that the forms of the graphs of Bernoulli type under the structure of product
of admissible graphs (see \cite{11}, \cite{12}) look like
the Feynman amplitudes very much. This similarity results in the existence
of one-one correspondence mentioned above. 

Moreover, as another application of our 
construction we discuss the various forms of 
Wick theorem, Wick power and expectation of Wick-monomial in terms of
the coordinates in $\mathbb{R}^{d}$ from the viewpoint of the star product
of functions. In the sense of the star product the Wick theorem, Wick power and
expectation of Wick-monomial for the case of scalar fields can be obtained from 
 their various forms mentioned above.
Observing the procedure to calculate the star product we find the Feynman
amplitudes arise from the bi-vector field, that explains also why the Kontsevich graphs
are relevant to the Feynman graphs.

This paper is organised as the following. In section 2 we discuss the star products at
level of functions. The definitions of star product and Poisson bracket are presented
and some properties are discussed. In section 3 we recall some contents of
the Kontsevich graphs including admissible graphs and their product,
Bernoulli graphs, et cetera(see \cite{11}, \cite{12}, \cite{13}). 
A combinatorial notation, adjacency matrix, is introduced.
In the end of this section we prove the existence of one-one correspondence
between the graphs of Bernoulli type and the Feynman graphs.
In section 4 we discuss the Wick theorem, Wick power and notion of expectation
of Wick-monomial from the viewpoint of star product of functions on $\mathbb{R}^{d}$.
Here everything is expressed in terms of functions, or special, coordinates on
$\mathbb{R}^{d}$. In section 5 we discuss the star products of the fields and
functionals based on the star product of functions on $\mathbb{R}^{d}$.

\section{The star products at level of functions}

In this section we discuss the star products of functions which plays the role of underlying structure 
about the star products of scalar fields and functionals, moreover, includes all of combinatorial and
algebraic information concerning the star products  of scalar fields almost.

\subsection{The star product with tensor form}

At first we introduce some notations. 
Let $\mathcal{A}$ be a commutative algebra over $\mathbb{R}$ (or $\mathbb{C}$)
with finite or countable generators, we consider a free $C^{\infty}(\mathbb{R}^{d})$ module on
$\mathcal{A}$ denoted by $C^{\infty}_{\mathcal{A}}$,

$$
C^{\infty}_{\mathcal{A}}(\mathbb{R}^{d})
=\bigoplus\limits _{i\in\mathcal{A}}C^{\infty}(\mathbb{R}^{d}).
$$
The elements in $C^{\infty}_{\mathcal{A}}$ are the linear combinations of the
elements in $\mathcal{A}$ with coefficients in $C^{\infty}(\mathbb{R}^{d})$.
The partial derivations on $C^{\infty}(\mathbb{R}^{d})$ can be extended to
$C^{\infty}_{\mathcal{A}}$ where the elements in $\mathcal{A}$ are viewed as constants, 
for example, we have $\partial_{i}(\lambda f(\textbf{x}))=\lambda\partial_{i}f(\textbf{x}),$
 $\lambda\in\mathcal{A}, f(\textbf{x})\in C^{\infty}(\mathbb{R}^{d})$,
 here we have used the short symbols, $\textbf{x}=(x_{1},\cdots,x_{d})$,
 $\partial_{i}=\partial_{x_{i}}$.
In the present article we focus on the situation of real scalar fields, the case of complex
ones are similar, thus we discuss the problems over real number field $\mathbb{R}$ below.

Now we consider the derivations of the tensor of functions. Let 
$f_{i}(\textbf{x})\in  C^{\infty}(\mathbb{R}^{d}),\,i=1,\cdots,m$.
We define the partial derivations for tensor of the functions 
$f_{1}(\textbf{x}_{1})\otimes\cdots\otimes f_{m}(\textbf{x}_{m})$
as following:
\begin{equation}
\partial_{i}^{(j)}(\bigotimes_{j=1}^{m}f_{j}(\textbf{x}_{j})) 
=f_{1}(\textbf{x}_{1})\otimes\cdots\otimes\partial_{i}
f_{j}(\textbf{x}_{j})\otimes
\cdots\otimes f_{m}(\textbf{x}_{m}),
\end{equation}
where the variables of $f_{j}(\textbf{x}_{j})$ 
($j=1,\cdots,m$) are denoted by 
$\textbf{x}_{j}=(x_{j}^{(1)},\cdots,x_{j}^{(d)})$.
$\partial_{i}^{(j)}$ acts on $j$th factor in above tensor
and it is the partial derivation with respect to
the $i$th component of $\textbf{x}_{j}$.

In this paper the star product at level of functions 
what we want to construct is Moyal-like
one. For simplicity, in the discussions below, we will
restrict our considerations in the situations of the functions
with one variable, i.e. each factor in the
tensor is a function with one variable.
Then, the formulas (2.1) will be of
the following forms:

\begin{equation}
\partial_{j}(\bigotimes_{j=1}^{m}f_{j}(\textbf{x}_{j})) 
=f_{1}(x_{1})\otimes\cdots\otimes f_{j}^{\prime}(x_{j})\otimes
\cdots\otimes f_{m}(x_{m}),
\end{equation}

Let $K=\{K_{ij}|K_{i,j}\in\mathcal{A},\,i,j\in\mathbb{Z}^{+}\}$,
then we have the definition of the star product as follows:  

\begin{definition}
Let $f_{i}(x_{i}),g_{j}(y_{j})\in C^{\infty}(\mathbb{R}),\,i=1,\cdots,m,\,j=1,\cdots,n$, their star product
with tensor form is defined by the following formula:

\begin{equation}
\begin{array}{ c}
   \left[(f_{1}(x_{1})\otimes\cdots\otimes f_{m}(x_{m}))\star_{K}
   (g_{1}(y_{1})\otimes\cdots\otimes g_{n}(y_{n}))\right]_{\otimes}  \\
 =\exp\{\hbar\mathcal{K}_{\textbf{x},\textbf{y}}\}
(f_{1}(x_{1})\otimes\cdots\otimes f_{m}(x_{m}))\otimes
   (g_{1}(y_{1})\otimes\cdots\otimes g_{n}(y_{n})), 
\end{array}
\end{equation}
where 

$$
\mathcal{K}_{\textbf{x},\textbf{y}}=
\sum_{ij}K_{i,m+j}\partial_{x_{i}}\otimes\partial_{y_{j}}.
$$
\end{definition}

\begin{remark}
$\\$
\begin{itemize}
\item In definition 2.1 the subscripts $i$ and $j$ indicate the positions
of the factors in the tensor, for example, index $j$ 
indicates $g_{j}(y_{j})$ which is
the $(m+j)$th factor in $f_{1}(x_{1})\otimes\cdots\otimes f_{m}(x_{m})\otimes
g_{1}(y_{1})\otimes\cdots\otimes g_{n}(y_{n})$.
Therefore $K_{i,m+j}$ concerns $i$th and $(m+j)$th factors.
\item Let

$$
\mathbf{C}_{\mathcal{A},\hbar}^{\infty}(\mathbb{R})
=\{\sum\limits_{n\geq 0}\hbar^{n}f_{n}(x)|
f_{n}(x)\in\mathbf{C}_{\mathcal{A}}^{\infty}(\mathbb{R})\}.
$$
Then the star product in definition 2.1 defines
a map

$$
\begin{array}{c}
\begin{array}{cccc}
\star_{K}: &
\underbrace{(\mathbf{C}^{\infty}(\mathbb{R})\otimes
	\cdots\otimes\mathbf{C}^{\infty}(\mathbb{R}))} & \times &
\underbrace{(\mathbf{C}^{\infty}(\mathbb{R})\otimes
	\cdots\otimes\mathbf{C}^{\infty}(\mathbb{R}))} \\
 & m-times &  & n-times
\end{array}
 \\  
\begin{array}{cc}
\longrightarrow &
 \underbrace{\mathbf{C}_{\mathcal{A},\hbar}^{\infty}(\mathbb{R})
	\otimes\cdots\otimes
	\mathbf{C}_{\mathcal{A},\hbar}^{\infty}(\mathbb{R})}.  \\
 & (m+n)-times 
\end{array}
\end{array}
$$
There is a natural and obvious way to extend the star 
product as a map 

$$
\begin{array}{cccc}
\star_{K}: &
\underbrace{\mathbf{C}_{\mathcal{A},\hbar}^{\infty}(\mathbb{R})
\otimes\cdots\otimes
\mathbf{C}_{\mathcal{A},\hbar}^{\infty}(\mathbb{R})} &
\longrightarrow &
\underbrace{\mathbf{C}_{\mathcal{A},\hbar}^{\infty}(\mathbb{R})
\otimes\cdots\otimes
\mathbf{C}_{\mathcal{A},\hbar}^{\infty}(\mathbb{R})}. \\
  & (m+n)-times &  & (m+n)-times
\end{array}
$$
\end{itemize}	
\end{remark}

To simplify the formula (2.3) we introduce short notations in the following way. 
Let $F_{m}(\textbf{x})=\bigotimes_{i}f_{i}(x_{i})$,
$G_{n}(\textbf{y})=\bigotimes_{j}g_{j}(y_{j})$.
Then the formula (2.3) can be denoted in a short way as following

\begin{equation}
      (F_{m}(\textbf{x})\star_{K}G_{n}(\textbf{y}))_{\otimes}    
      =\exp\{\hbar\mathcal{K}_{\textbf{x},\textbf{y}}\}(F_{m}(\textbf{x})\otimes G_{n}(\textbf{y})). 
\end{equation}

With the help of the formula (2.3) or (2.4), 
by a straightforward computation we know that the associativity is valid, i.e. we have

\begin{equation}
\begin{array}{cc}
      &\left[(F_{m}(\textbf{x})\star_{K}G_{n}(\textbf{y}))_{\otimes}\star_{K}H_{p}(\textbf{z})\right]_{\otimes}    \\
      =&\left[F_{m}(\textbf{x})\star_{K} (G_{n}(\textbf{y})\star_{K}H_{p}(\textbf{z}))_{\otimes}\right] _{\otimes} ,
\end{array}
\end{equation}
where $H_{p}(\textbf{z})=\bigotimes_{k=1}^{p}h_{k}(z_{k})$.
Actually, 

$$
\begin{array}{c}
\left[(F_{m}(\textbf{x})\star_{K}G_{n}(\textbf{y}))_{\otimes}
\star_{K}H_{p}(\textbf{z})\right]_{\otimes} \\
=\exp\{\hbar\mathcal{K}_{\textbf{x},\textbf{z}}
+\hbar\mathcal{K}_{\textbf{y},\textbf{z}}\}
\left(\exp\{\hbar\mathcal{K}_{\textbf{x},\textbf{y}}\}
(F_{m}(\textbf{x})\otimes G_{n}(\textbf{y}))\otimes
H_{p}(\textbf{z}) 
\right) \\
=\exp\{\hbar(\mathcal{K}_{\textbf{x},\textbf{y}}+
\mathcal{K}_{\textbf{x},\textbf{z}}+
\mathcal{K}_{\textbf{y},\textbf{z}})\}
(F_{m}(\textbf{x})\otimes G_{n}(\textbf{y})\otimes H_{p}(\textbf{z})).
\end{array}
$$
On the other hand,

$$
\begin{array}{c}
\left[F_{m}(\textbf{x})\star_{K} (G_{n}(\textbf{y})
\star_{K}H_{p}(\textbf{z}))_{\otimes}\right] _{\otimes} \\
=\exp\{\hbar(\mathcal{K}_{\textbf{x},\textbf{y}}+
\mathcal{K}_{\textbf{x},\textbf{z}}+
\mathcal{K}_{\textbf{y},\textbf{z}})\}
(F_{m}(\textbf{x})\otimes G_{n}(\textbf{y})\otimes H_{p}(\textbf{z})).
\end{array}
$$

Thus, the both sides of the formula (2.5) are of the following form
 
\begin{equation}
\exp\{\hbar(\mathcal{K}_{\textbf{x},\textbf{y}}+\mathcal{K}_{\textbf{x},\textbf{z}}+\mathcal{K}_{\textbf{y},\textbf{z}})\}
(F_{m}(\textbf{x})\otimes G_{n}(\textbf{y})\otimes H_{p}(\textbf{z})), 
\end{equation}
Explicitly, 

$$
\begin{array}{c}
\mathcal{K}_{\textbf{x},\textbf{y}}=
\sum_{ij}K_{i,m+j}\partial_{x_{i}}\otimes\partial_{y_{j}}\otimes id_{p}, \\
\mathcal{K}_{\textbf{x},\textbf{z}}=
\sum_{ik}K_{i,m+n+k}\partial_{x_{i}}
\otimes id_{n}\otimes\partial_{z_{k}}, \\
\mathcal{K}_{\textbf{y},\textbf{z}}=
\sum_{jk}K_{m+j,m+n+k}id_{m}\otimes\partial_{y_{j}}
\otimes\partial_{z_{k}}.
\end{array}
$$

Similar to the discussion in remark 2.1,
the subscripts $\textbf{x}$, $\textbf{y}$
and $\textbf{z}$ indicate the positions of the 
factors in the tensor.

\begin{remark}
If each factor of the tensor is a function 
in $\mathbf{C}^{\infty}(\mathbb{R}^{d})$,
in addition one needs to consider 
the positions of the components
of the variables. The star product can be generalized
to this situation in a natural way.
\end{remark}

\subsection{The ordinary star products }

Paralleling to definition 2.1 
we can define the star product in ordinary
sense. 

\begin{definition}
Let $f(x_{1},\cdots,x_{m})\in\mathbf{C}^{\infty}(\mathbb{R}^{m})$,
$g(y_{1},\cdots,y_{n})\in\mathbf{C}^{\infty}(\mathbb{R}^{n})$,
then

\begin{equation}
\begin{array}{c}
f(x_{1},\cdots,x_{m})\star_{K}g(y_{1},\cdots,y_{n}) \\
=\exp\{\mathcal{K}_{\textbf{x},\textbf{y}}\}
(f(x_{1},\cdots,x_{m})g(y_{1},\cdots,y_{n})),
\end{array}
\end{equation}
where

$$
\mathcal{K}_{\textbf{x},\textbf{y}}=\sum\limits_{i,j}K_{i,m+j}
\partial_{x_{i}}\partial_{y_{j}}.
$$
\end{definition}

It is easy to check that

\begin{equation}
(\textbf{m}\circ F_{m}(\textbf{x}))\star_{K}(\textbf{m}\circ G_{n}(\textbf{y}))
=\textbf{m}\circ(F_{m}(\textbf{x})\star_{K}G_{n}(\textbf{y}))_{\otimes},
\end{equation}
where $\textbf{m}$ denotes the point-wise multiplication of the functions,

$$
\textbf{m}:f_{1}(x_{1})\otimes\cdots\otimes f_{m}(x_{m})
\mapsto f_{1}(x_{1})\cdots f_{m}(x_{m}).
$$
Actually, by a straightforward calculation
one can check that

$$
(\textbf{m}\circ F_{m}(\textbf{x}))\star_{K}
(\textbf{m}\circ G_{n}(\textbf{y}))
=\sum\limits_{k\geq 0}\frac{\hbar^{k}}{k!}
(\sum\limits_{1\leq i\leq m,1\leq j\leq n}
K_{i,m+j}\partial_{x_{i}}\partial_{y_{j}})^{k}
\textbf{m}\circ F_{m}(\textbf{x})
\textbf{m}\circ G_{n}(\textbf{y}).
$$

From the formula (2.8), it is easy for us to check that
$$
\begin{array}{cc}
  &\left[(\textbf{m}\circ F_{m}(\textbf{x}))\star_{K}(\textbf{m}\circ G_{n}(\textbf{y}))\right]\star_{K} (\textbf{m}\circ H_{p}(\textbf{z})) \\
 =&\textbf{m}\circ((F_{m}(\textbf{x})\star_{K}G_{n}(\textbf{y}))_{\otimes})\star_{K}(\textbf{m}\circ H_{p}(\textbf{z}))  \\
=&\textbf{m}\circ\left[ ((F_{m}(\textbf{x})\star_{K}G_{n}(\textbf{y}))_{\otimes}\star_{K} H_{p}(\textbf{z})\right]_{\otimes},
\end{array}
$$
where $H_{p}(\textbf{z})$ is as above. Therefore, 
the associativity of the star product
defined in definition 2.2 is valid.
Furthermore we have

\begin{definition}
Let $F_{m}(\textbf{x}),G_{m}(\textbf{x})$
be as above, we define

\begin{equation}
(\textbf{m}\circ F_{m}(\textbf{x}))\star_{K}(\textbf{m}\circ G_{m}(\textbf{x}))
=(\textbf{m}\circ F_{m}(\textbf{x}))\star_{K}(\textbf{m}\circ G_{n}(\textbf{y}))|_{\textbf{x}=\textbf{y}}.
\end{equation}
\end{definition}

\begin{remark}
$\\$
\begin{itemize}
  \item Comparing with the Moyal product
$$ f(\textbf{x})\star g(\textbf{x})=f(\textbf{x})g(\textbf{x})+\hbar\sum\limits_{i,j}\alpha^{ij}
     \partial_{i}f(\textbf{x})\partial_{j}g(\textbf{x})+O(\hbar^{2}),$$
in the present paper the constant coefficients $\alpha^{ij}$ in Moyal star product are replaced by abstract elements $K_{ij}$
in $\mathcal{A}.$ However, similar to the case of the Moyal product, $\mathcal{K}$
plays the role of bi-vector field with coefficients in $\mathcal{A}$.  
  \item It is obvious that the star product defined by the formula (2.9) can be extended to 
  the case of $C^{\infty}_{\mathcal{A},\,\hbar}$, where
 $$C^{\infty}_{\mathcal{A},\,\hbar}=\{\sum\limits_{m\geqslant 0}\hbar^{m}F_{m}|
 F_{m}\in C^{\infty}_{\mathcal{A}},m\geqslant0\},$$ 
such that the star product $\star_{K}$ is a map from $C^{\infty}_{\mathcal{A},\,\hbar}\times C^{\infty}_{\mathcal{A},\,\hbar}$
to $C^{\infty}_{\mathcal{A},\,\hbar}$, or from $C^{\infty}_{\mathcal{A},\,\hbar}\otimes C^{\infty}_{\mathcal{A},\,\hbar}$
into itself.
  \item We can extend the star product defined as above to the case where functions depending 
  on some parameters. For example,
  $$
   f(t,\textbf{x})\star_{K}g(t,\textbf{y})=\exp\{\hbar\mathcal{K}\}( f(t,\textbf{x}) g(t,\textbf{y})),
  $$
where $t=(t_{1},\cdots,t_{k})$.
\end{itemize}
\end{remark}

More precisely, we have

$$
\begin{array}{c}
\textbf{m}\circ F_{m}(\textbf{x})\star_{K} \textbf{m}\circ G_{n}(\textbf{y}) \\
=\textbf{m}\circ F_{m}(\textbf{x})\textbf{m}\circ G_{n}(\textbf{y})
+\hbar\sum\limits_{i,j}K_{i,m+j}\partial_{i}
(\textbf{m}\circ F_{m}(\textbf{x}))
\partial_{j}(\textbf{m}\circ G_{n}(\textbf{y}))    \\
+\frac{\hbar^{2}}{2} \sum\limits_{i_{1},i_{2},j_{1},j_{2}}
K_{i_{1},m+j_{1}}K_{i_{2},m+j_{2}} 
\partial_{i_{1}}\partial_{i_{2}}(\textbf{m}\circ F_{m}(\textbf{x}))
\partial_{j_{1}}\partial_{j_{2}}
(\textbf{m}\circ G_{n}(\textbf{y}))+\cdots,
\end{array}
$$
thus we have

$$
\begin{array}{c}
\textbf{m}\circ F_{m}(\textbf{x})\star_{K}\textbf{m}\circ G_{n}(\textbf{y})
-\textbf{m}\circ G_{n}(\textbf{y})\star_{K}\textbf{m}\circ F_{m}(\textbf{x}) \\
=\hbar\sum\limits_{i\neq j}(K_{i,m+j}-K_{j,m+i})
\partial_{i}(\textbf{m}\circ F_{m}(\textbf{x}))\partial_{j}
(\textbf{m}\circ G_{n}(\textbf{y}))+O(\hbar^{2}).
\end{array}
$$
If 

$$
\textbf{m}\circ F_{m}(\textbf{x})\star_{K}\textbf{m}\circ G_{m}(\textbf{y})
=\textbf{m}\circ G_{m}(\textbf{y})\star_{K}\textbf{m}\circ F_{m}(\textbf{x})
$$
we say the star product $\star_{K}$ is commutative. It is obvious that we have
\begin{proposition}
The star products (2.9) are commutative iff the propagator matrix 
$K_{i,n+j}=K_{j,n+i}$.
\end{proposition}

For non-commutative case we define the Poisson bracket as following:

\begin{equation}
\{\textbf{m}\circ F_{m}(\textbf{x}),\textbf{m}\circ G_{m}(\textbf{x})\}_{K}
=\sum\limits_{i\neq j}(K_{i,m+j}-K_{j,m+i})
\partial_{i}(\textbf{m}\circ F_{m}(\textbf{x}))
\partial_{j}(\textbf{m}\circ G_{m}(\textbf{x})).
\end{equation}
The Poisson bracket can be extended to $C^{\infty}_{\mathcal{A}}$ also.
It is obvious that the Poisson brackets (2.10) is bi-linear and anti-symmetric,
additionally, are derivations for both of 
$\textbf{m}\circ F_{m}$ and $\textbf{m}\circ G_{m}$. The Jacobi
identity is valid for the Poisson brackets, i.e. we have:

$$
\{\{\textbf{m}\circ F_{m}(\textbf{x}),\textbf{m}\circ G_{m}(\textbf{x})\}
,\textbf{m}\circ H_{m}(\textbf{x})\}+\emph{cycles}=0.
$$

Now we extend the star product to the situation with several factors. We have
\begin{proposition}
Let $f_{i}(x_{i})\in C^{\infty}(\mathbb{R}),\,i=1,\cdots,m$,
we have
\begin{equation}
f_{1}(x_{1})\star_{K}\cdots\star_{K} f_{m}(x_{m}) 
=\textbf{m}\circ(f_{1}(x_{1})\star_{K}\cdots\star_{K} f_{m}(x_{m}))_{\otimes}
\end{equation}
where
$$
\begin{array}{c}
(f_{1}(x_{1})\star_{K}\cdots\star_{K} f_{m}(x_{m}))_{\otimes}   \\
=\exp\{\hbar\sum\limits_{i<j}\mathcal{K}_{ij}\}(f_{1}(x_{1})\otimes\cdots\otimes f_{m}(x_{m})),
\end{array}
$$
$\mathcal{K}_{ij}=K_{ij}\partial_{x_{i}}\otimes\partial_{x_{j}}$,
$\partial_{x_{k}}$ acts on $f_{k}(x_{k})$ as ordinary partial derivation, 
$1\leqslant k\leqslant m$.
\end{proposition}

Observing the definition 2.1 we know that the different choice of propagator matrixes
determines the different star products. Now we discuss the connection between the star products
depending on the different propagator matrixes. We have
\begin{theorem}
For different propagator matrixes $K,K^{\prime}$ we have
\begin{equation}
(\textbf{m}\circ F_{m}(\textbf{x})\star_{K} \textbf{m}\circ G_{n}(\textbf{y}))_{\otimes}
=\exp\{\hbar(\mathcal{K}-\mathcal{K}^{\prime})\}
(\textbf{m}\circ F_{m}(\textbf{x})\star_{K^{\prime}} 
\textbf{m}\circ G_{n}(\textbf{y}))_{\otimes}.
\end{equation}
\end{theorem}
\begin{proof}
Actually we have
$$\exp\{\hbar\mathcal{K}\}=\exp\{\hbar(\mathcal{K}-\mathcal{K}^{\prime})\}\exp\{\mathcal{K}^{\prime}\}.$$
\end{proof}
\begin{remark}
It is obvious that all formulas in above discussion are valid for elements in $C^{\infty}_{\mathcal{A},\,\hbar}$.
Thus we will only discuss the issues concerning the star products for smooth functions below. 
\end{remark}

\section{Kontsevich graphs and Feynman graphs}

In this section we will explain how the star products result in the Feynman amplitudes,
at same time, Kontsevich graphs result in the Feynman graphs naturally.

\subsection{Notations}

At beginning as preparations we recall some contents concerning the Kontsevich graphs simply, 
for more details we direct readers to \cite{13}.

\begin{definition}
An oriented graph  $\Gamma$ is a pair $(V_{\Gamma},E_{\Gamma})$ of two finite
sets such that $E_{\Gamma}$ is a subset of $V_{\Gamma}\times V_{\Gamma}$.
 Elements of $V_{\Gamma}$ are vertices of $\Gamma$, elements of $E_{\Gamma}$
 are edges of $\Gamma$. If $e=(v_{1},v_{2})\in E_{\Gamma}\subseteq V_{\Gamma}\times V_{\Gamma}$ 
 is an edge of $\Gamma$ then we say that $e$ stars at $v_{1}$ and ends at $v_{2}$.
\end{definition}

\begin{definition}
(\textbf{Admissible graphs}, \cite{13}, p.22) Admissible graph 
$G_{n,m}$ is an oriented graph with labels such that
\begin{itemize}
  \item The set of vertices $V_{\Gamma}$ is 
  $\{v_{1}^{(1)},\cdots,v_{n}^{(1)}\}\sqcup\{v_{1}^{(2)},\cdots,v_{m}^{(2)}\}$ 
  where $n,m\in \mathbb{Z}_{\geqslant 0},\, 2m+n-2\geqslant 0$; 
  vertices from $\{v_{1}^{(1)},\cdots,v_{n}^{(1)}\}$
  are called vertices of the first type, 
  vertices from $\{v_{1}^{(2)},\cdots,v_{m}^{(2)}\}$ are called vertices
  of the second type.
  \item Every edge $e=(v_{1},v_{2})\in E_{\Gamma}$ stars at a vertex of the first type,
  $v_{1}\in \{v_{1}^{(1)},\cdots,v_{n}^{(1)}\}$.
  \item There are no loops, i.e. no edges of the type $(v,v)$.
  \item For every vertex $k\in\{1,\cdots,n\}$ of the first type, the set of edges
  $$Star(k)=\{(v_{1},v_{2})\in E_{\Gamma}|v_{1}=k\}$$
  starting from $k$, is labeled by symbols $\{e_{k}^{1},\cdots,e_{k}^{\#Star(k)}\}$.
\end{itemize}
\end{definition}

In other articles the vertices of the first type are also called internal vertices
and vertices of the second type are called boundary vertices. About the operation
of graphs we have

\begin{definition}
 (see  \cite{11}, p.7 and \cite{12}, p.22)If $\Gamma_{1}\in G_{n,m}$, $\Gamma_{2}\in G_{n^{\prime},m}$,
 we define the product $\Gamma_{1}\Gamma_{2}\in G_{n+n^{\prime},\,m}$ as the graph obtained
 from disjoint union of two graphs by identification of the vertices of the second type. 
We call this product Kathotia product. We define $\Gamma\emptyset=\Gamma$.
\end{definition} 

It is easy to see that the product of admissible graphs defined above is commutative.
For convenience we define the embedding of the admissible graphs $G_{n,m}\hookrightarrow G_{n,m^{\prime}}$,
$m^{\prime}\geqslant m$, or, extend a graph in $G_{n,m}$ to a graph in $G_{n,m^{\prime}}$.

\begin{definition}
Let $\Gamma\in G_{n,m}$, the all vertices of the second type 
$\Gamma$ labeled by $v_{1}^{(2)},\cdots,v_{m}^{(2)}$,
for a subset of $I=\{i_{1},\cdots,i_{m}\}\subseteq\{1,\cdots,m^{\prime}\}$, 
$1\leqslant \bar{i}_{1}<\cdots<\bar{i}_{m}\leqslant \bar{m^{\prime}}$, we extend $\Gamma$ in the following
way:
\begin{itemize}
  \item identifying vertex $v_{j}^{(2)}$ with vertex $v_{i_{j}}^{(2)}$, $j=1,\cdots,m$.
\end{itemize}
Above procedure define an embedding $\iota_{I}:G_{n,m}\hookrightarrow G_{n,m^{\prime}}$,
and we denote new graph by $\Gamma_{I}$.
We call $I$ the position of $\iota_{I}$
or $\iota_{I}(\Gamma)$.
\end{definition} 

\begin{remark}
Combining the definitions 3.3 and 3.4 we can consider the product of general admissible graphs. 
For instance, let $\Gamma\in G_{n,m}$,
$\Gamma^{\prime}\in G_{n^{\prime},m^{\prime}}$, we take $m_{1}=\max\{m,m^{\prime}\}$
and choose two positions $(i_{1},\cdots,i_{m})$, $(i^{\prime}_{1},\cdots,i^{\prime}_{m^{\prime}})$,
then the product 
$$
\Gamma_{(i_{1},\cdots,i_{m})}\Gamma^{\prime}_{(i^{\prime}_{1},\cdots,i^{\prime}_{m^{\prime}})}
\in G_{n+n^{\prime},m_{1}}
$$
makes sense.
\end{remark}

\subsection{Kontsevich's rule for star product with tensor form}

\paragraph{Bernoulli graphs}

Now we embed the star product defined in definition 2.1 into the Kontsevich's framework. 
The Moyal product is one generated by constant Poisson bi-vector field, it may be trivial
case from the viewpoint of Kontsevich theory. Similar to the case of the Moyal product, 
the basic graph is Bernoulli graph (see \cite{11}, \cite{12})
$b\in G_{1,2}$ which endows one vertex of the first type, two vertices of the second type
named by left and right ones respectively,
and two edges starting at the unique vertex of the first type,
ending at left (or right) vertex of the second type denoted by $e^{L}$
and $e^{R}$ respectively. We call $b$ the basic Bernoulli graph. We consider

$$
\begin{array}{cc}
      b^{n}=& \underbrace{b\cdots b}, \\
      & n-times
\end{array}
$$
where the product of the graphs is Kathotia product in definition 3.3, 
thus we know that $b^{n}\in G_{n,2}$. $b^{n}$ is simple graph even in $G_{n,2}$,
actually, there are no edges connecting the different vertices of the first type in $b^{n}$. 
Here we do not distinguish any two vertices of the
first type in $b^{n}$. In fact, if
the vertices of the first type in $b^{n}$ are labeled by $\{1,\cdots,n\}$,
and we make a permutation of indices, 
the new graph is isomorphic to $b^{n}$.
Each vertex of the first type of $b^{n}$ assigns two edges starting
at this vertex and ending at different verteies of 
the second type. The edges of $b^{n}$
are labeled by symbols $e_{1}^{L},e_{1}^{R},\cdots, e_{n}^{L},e_{n}^{R}$,
where $e_{k}^{L}$ (or $e_{k}^{R}$)
denote the edge starting at $k$-th 
vertex of the first type and ending at left (or right) vertex of the second type.
It is obvious that each vertex of the first type
corresponds to a pair $(e^{L}_{k},e^{R}_{k})$.

We now turn to a general situation $G_{n,m}$,
which means the Kontsevich's graphs
with $m$ vertiecs of the second type.
A Bernoulli graph with $m$ vertices of
the second type has two features.
There are no edges connecting two
vertices of the first type. For each
vertex of the first type, there just be
two edges starting at this vertex and
ending at different vertices of the
second type. We indicate the $m$ vertices
of the second type by $\{1,\cdots,m\}$,
or, the set of the vertices of the second type
is denoted by $\{v_{1},\cdots,v_{m}\}$.
Then, with the help of Kathotia product,
any Bernoulli graph can be
generated by the Bernoulli graph $b_{ij}$
($1\leq i<j\leq m$) which is endowed with an
unique vertex of the first type and
two vertices of the second type, there are
two edges of $b_{ij}$ starting at the unique vertex 
of the first type and ending at ith and jth
vertices of the second type respectively.
Actually, $b_{ij}$ can be regarded as
the embedding of $b$ with position $\{i,j\}$,
that is $b_{ij}=\iota_{\{i,j\}}b$.
Let $\Gamma\in G_{n,m}$ be a Bernoulli graph, it is easy
to check that

\begin{equation}
\Gamma=\prod\limits_{1\leq i<j\leq m}
b^{m_{ij}}_{ij}.
\end{equation}
In the above formula, the product is
Kathotia product, $m_{ij}\in \mathbb{N}$
($1\leq i<j\leq m$) with $\sum_{ij}m_{ij}=n$,
when $m_{ij}=0$ $b_{ij}^{m_{ij}}=\emptyset$.
Let $\sigma\in\mathbb{P}_{m}$ be permutation
on $\{1,\cdots,m\}$. We define

$$
\sigma(\Gamma)=\prod\limits_{1\leq i<j\leq m}
b^{m_{\sigma(i)\sigma(j)}}_{\sigma(i)\sigma(j)}.
$$
Similar to the discussion in the situation of $b^{n}$,
for fixed $i$ and $j$ ($1\leq i,j\leq m$), we do not
distinguish any two vertices of the first type
which connect to both $i$th and $j$th vertices of the second
type. Under this assumption, it is obvious
that for two Bernoulli graphs $\Gamma_{1},\Gamma_{2}\in G_{n,m}$,
$\Gamma_{1}$ is isomorphic to $\Gamma_{2}$
denoted by $\Gamma_{1}\sim\Gamma_{2}$
if and only if
there is a $\sigma\in\mathbb{P}_{m}$ such that
$\Gamma_{1}=\sigma(\Gamma_{2})$.
We know $\sim$ is an equivalent relation.

Let 

$$
B_{n,m}=\{\Gamma\in G_{n,m}|\Gamma\,\,is\,\,a\,\,Bernoulli\,\,graph\},\,\,
B_{m}=(\bigcup\limits_{n}B_{n,m})\cup\{\emptyset\}.
$$
Then, we have

\begin{proposition}
Under Kathotia product $B_{m}$ 
is a monoid with generators
$\{b_{ij}\}_{1\leq i<j\leq m}$.
\end{proposition}

 It is obvious that
$B_{m}/\sim$ is a monoid also.

\begin{corollary}
$\mathbf{Span}_{\mathbb{R}}(B_{m})$ 
(or $\mathbf{Span}_{\mathbb{C}}(B_{m})$) 
is an algebra over $\mathbb{R}$ 
(or $\mathbb{C}$) with generators
$\{b_{ij}\}_{1\leq i<j\leq m}$.
\end{corollary}

\paragraph{Adjacency matrices}

We now turn to the adjacency matrices.

\begin{definition}
A adjacency matrix is a symmetric matrix with 
non-negative integer entries. Here we make an additional
restriction such that the entries on 
main diagonal are zeros, i.e. for an adjacency
matrix $M=(m_{ij})$ of order $m$, we have $m_{ij}\in\mathbb{N}$
$m_{ij}=m_{ji},\,m_{ii}=0$ ($1\leq i,j\leq m$). 
We call $\frac{1}{2}\sum_{ij}m_{ij}$ the
degree of $M$ denoted by $degM$. Let
$M_{adj}(m,\mathbb{N})$ denote the set of all 
adjacency matrices of order $m$.
\end{definition}

Let $M(i,j)=(m_{kl})_{m\times m}$ satisfying
$m_{kl}=\delta_{ik}\delta_{jl}$, where
$i<j,\,k<l$. Then $M(i,j)$ is a permutation matrix 
which is in $M_{adj}(m,\mathbb{N})$ obviously.

\begin{proposition}
$M_{adj}(m,\mathbb{N})$ is a monoid with generators
$\{M(i,j)\},\,1\leq i<j\leq m$.
\end{proposition}

\begin{proof}
By definition 3.5 we know $0\in M_{adj}(m,\mathbb{N})$.
It is ovious that for $M_{1},M_{2}\in M_{adj}(m,\mathbb{N})$,
$M_{1}+M_{2}\in M_{adj}(m,\mathbb{N})$. On the other hand,
for each $M\in M_{adj}(m,\mathbb{N})$, we have
$M=\sum_{i,j}m_{ij}M(i,j)$. Because $m_{ij}\in\mathbb{N}$,
we know that

$$
\begin{array}{cc}
m_{ij}M(i,j)= & \underbrace{M(i,j)+\cdots+M(i,j)}. \\
   & m_{ij}-times
\end{array}
$$
Up to now we have 
proved the conclusion.
\end{proof}

Let $M_{1},M_{2}\in M_{adj}(m,\mathbb{N})$,
we say $M_{1}\sim M_{2}$, if there exists
a permutation matrix $P$ of order $m$ such
that $M_{1}=PM_{2}P$. $\sim$ is an
equivalent relation obviously. It is easy to check that
$M_{adj}(m,\mathbb{N})/\sim$ is a monoid also.

Starting from a given adjacency matrix 
$M=(m_{ij})\in M_{adj}(m,\mathbb{N})$ 
with $degM=n$, one can construct a Bernoulli graph
$\Gamma_{M}=\prod_{1\leq i<j\leq m}b_{ij}^{m_{ij}}$
with $n$ vertices of the first type and
$m$ vertices of the second type.
The previous construction results in a 
map 

\begin{equation}
\mathcal{A}_{\mathcal{B}}:M_{adj}(m,\mathbb{N})
\longrightarrow B_{m};\,
\mathcal{A}_{\mathcal{B}}(M)=\Gamma_{M},\,
M\in M_{adj}(m,\mathbb{N}).
\end{equation}
$\mathcal{A}_{\mathcal{B}}$ is an one-one
correspondence between the adjacency matrices
of order $m$
and Bernoulli graphs with $m$ vertices of
the second type.

It is easy to check that

\begin{equation}
\Gamma_{M_{1}+M_{2}}=
\Gamma_{M_{1}}\Gamma_{M_{2}},\,
M_{1},M_{2}\in M_{adj}(m,\mathbb{N}).
\end{equation}

In summary, we reach the following
conclosion:

\begin{proposition}
The following maps

\begin{itemize}
\item 

$$
\mathcal{A}_{\mathcal{B}}:
M_{adj}(m,\mathbb{N})\longrightarrow B_{m},
$$
\item

$$
\mathcal{A}_{\mathcal{B}}:
\mathbf{Span}_{\mathbb{K}}(M_{adj}(m,\mathbb{N}))
\longrightarrow\mathbf{Span}_{\mathbb{K}}(B_{m}),
\,\,\mathbb{K}=\mathbb{R}\,\,or\,\,\mathbb{C}, 
$$
\end{itemize}
are homomorphisms.
\end{proposition}

\begin{remark}
$\mathcal{A}_{\mathcal{B}}$ induces
a homomorphism $\mathcal{A}_{\mathcal{B}}:
M_{adj}(m,\mathbb{N})/\sim
\longrightarrow B_{m}/\sim$ also.
\end{remark}

\paragraph{Kontsevich's rule}

We now begin to discuss the Kontsevich's rule under
our consideration. Here the star product works
on the tensor of the functions space as the following

$$
\begin{array}{c}
\underbrace{\mathbf{C}^{\infty}_{\hbar}(\mathbb{R})
\otimes\cdots\otimes\mathbf{C}^{\infty}_{\hbar}(\mathbb{R})}. \\
m-times
\end{array}
$$
If we indicate the factors in above tensor
by $\{1,\cdots,m\}$ according to the
order from left to right, then the $i$th
factor corresponds to the $i$th vertex
of the second type of the Bernoulli
graphs in $B_{m}$. Let
$\mathcal{K}=\{K_{ij}|K_{ij}\in\mathcal{A},
1\leq i<j\leq m\}$, $\partial_{i}$ denote
the derivation acting on $i$th factor
in the above tensor.

Under the setting of this paper, it is enough
to consider Bernoulli graphs when we discuss
Kontsevich's rule. Recalling the previous discussions
concerning Bernoulli graphs, $B_{m}$ 
is a monoid with generators $\{b_{ij}\}$,
it is sufficient for us to define Kontesvich's
rule on $b_{ij}$ ($1\leq i<j\leq m$), that is

\begin{equation}
\mathcal{U}(b_{ij},\mathcal{K})
=\mathcal{K}_{ij},
\,1\leq i<j\leq m,
\end{equation}
where $\mathcal{K}_{ij}=K_{ij}\partial_{i}\otimes\partial_{j}$
which is same as one in proposition 2.2.
For the general Bernoulli graphs, for example,
$\Gamma=\sum_{1\leq i<j\leq m}b_{ij}^{m_{ij}}\in B_{m}$,
Kontsevich's rule can be extended in the
following way,

\begin{equation}
\mathcal{U}(\Gamma,\mathcal{K})=
\prod\limits_{1\leq i<j\leq m}
\mathcal{U}(b_{ij},\mathcal{K})^{m_{ij}}
=\prod\limits_{1\leq i<j\leq m}
\mathcal{K}_{ij}^{m_{ij}}.
\end{equation}
More precisely,

$$
\prod\limits_{1\leq i<j\leq m}\mathcal{K}_{ij}^{m_{ij}}=
\prod\limits_{1\leq i<j\leq m}K_{ij}^{m_{ij}}
\partial_{1}^{\alpha_{1}}\otimes\cdots
\otimes\partial_{m}^{\alpha_{m}},
$$
where $\alpha_{1}=\sum_{j>1}m_{1j}$, 
$\alpha_{m}=\sum_{i<m}m_{im}$,
$\alpha_{i}=\sum_{j<i}m_{ji}+
\sum_{j>i}m_{ij}$ ($1<i<m$).
The previous discussions result
in the following conclusion:

\begin{proposition}
Let $\Gamma_{1},
\Gamma_{2}\in B_{m}$, we have

\begin{equation}
\mathcal{U}(\Gamma_{1}\Gamma_{2},\mathcal{K})=
\mathcal{U}(\Gamma_{1},\mathcal{K})
\mathcal{U}(\Gamma_{2},\mathcal{K}).
\end{equation}
\end{proposition}

Let

$$
D_{\mathcal{A},m}=
\{\sum\limits_{|\alpha|\leq k}a_{\alpha}
\partial_{1}^{\alpha_{1}}\otimes\cdots\otimes
\partial_{m}^{\alpha_{m}}|a_{\alpha}\in\mathcal{A},\,
k\in\mathbb{N}\},
$$
where $\alpha=(\alpha_{1},\cdots,\alpha_{m})\in \mathbb{N}^{m}$
are multiple indices, $|\alpha|=\alpha_{1}+\cdots+\alpha_{m}$.
Then the formula (3.6) means that Kontsevich's rule
induces a homomorphism

\begin{equation}
\mathcal{U}(\cdot,\mathcal{K}):
\mathbf{Span}_{\mathbb{R}}(B_{m})\rightarrow D_{\mathcal{A},m}.
\end{equation}
Furthermore

\begin{equation}
\mathcal{U}(\cdot,\mathcal{K})\circ
\mathcal{A}_{\mathcal{B}}:
M\mapsto \mathcal{K}_{M},
\end{equation}
where $\mathcal{K}_{M}=\prod_{1\leq i<j\leq m}
\mathcal{K}_{ij}^{m_{ij}}$.

Now we introduce the generating function of
the Bernoulli graphs in $B_{m}$ with form
as the following:

$$
g_{B_{m}}(t)=\exp\{t\sum_{1\leq i<j\leq m}b_{ij}\}.
$$
By a straightforward calculation we get

\begin{lemma}
\begin{equation}
\exp\{t\sum_{1\leq i<j\leq m}b_{ij}\}=
\sum\limits_{M\in M_{adj}(m,\mathbb{N})}
\frac{t^{degM}}{M!}\Gamma_{M},
\end{equation}
where $M!=\prod_{1\leq i<j\leq m}m_{ij}!$.
\end{lemma}

\begin{proof}
Noting that

$$
\exp\{t\sum_{1\leq i<j\leq m}b_{ij}\}=
\sum\limits_{k\geq 0}\frac{t^{k}}{k!}
(\sum_{1\leq i<j\leq m}b_{ij})^{k},
$$
and

$$
(\sum_{1\leq i<j\leq m}b_{ij})^{k}=
\sum\limits_{m_{ij},\sum m_{ij}=k,i<j}
\frac{k!}{\prod_{i,j}m_{ij}!} 
\prod\limits_{1\leq i<j\leq m}
b_{ij}^{m_{ij}},
$$
we have

$$
\exp\{t\sum_{1\leq i<j\leq m}b_{ij}\}=
\sum\limits_{k\geq 0}t^{k}
\sum\limits_{M,degM=k}
\frac{1}{M!}\Gamma_{M}=
\sum\limits_{M\in M_{adj}(m,\mathbb{N})}
\frac{t^{degM}}{M!}\Gamma_{M}.
$$

\end{proof}

Recalling Kontsevich's rule 
$\mathcal{U}$ is a homomorphism, we have

$$
\mathcal{U}(\exp\{\hbar\sum_{1\leq i<j\leq m}b_{ij}\},\mathcal{K}) 
=\sum\limits_{M\in M_{adj}(m,\mathbb{N})}
\frac{\hbar^{degM}}{M!}
\mathcal{U}(\Gamma_{M},\mathcal{K}).
$$
Furthermore, combining with the formula (3.4)
we have

$$
\begin{array}{c}
\mathcal{U}(\exp\{\hbar\sum_{1\leq i<j\leq m}b_{ij}\},\mathcal{K}) \\
=\sum\limits_{M\in M_{adj}(m,\mathbb{N})}
\frac{\hbar^{degM}}{M!}\mathcal{U}(\Gamma_{M},\mathcal{K}) \\
=\sum\limits_{M\in M_{adj}(m,\mathbb{N})}
\frac{\hbar^{degM}}{M!}\mathcal{K}_{M} \\
=\exp\{\hbar\sum\limits_{1\leq i<j\leq m}\mathcal{K}_{ij}\}.
\end{array}
$$

In summary, with the help of Kontsevich's rule,
the star product in proposition 2.2 can 
be expressed in the following way.

\begin{theorem}
\begin{equation}
\begin{array}{c}
\mathcal{U}(\exp\{\hbar\sum_{1\leq i<j\leq m}b_{ij}\},\mathcal{K}) \\
=\exp\{\hbar\sum\limits_{1\leq i<j\leq m}\mathcal{K}_{ij}\} 
=\sum\limits_{M\in M_{adj}(m,\mathbb{N})}
\frac{\hbar^{degM}}{M!}\mathcal{K}_{M}.
\end{array}
\end{equation} 
\end{theorem}

\paragraph{From Kontsevich graphs to Feynman graphs}

Now we turn to the Feynman amplitudes and Feynman graphs. 
Firstly, we consider the multiple star product.
Recalling proposition 2.2 we know that

\begin{equation}
\begin{array}{c}
        (f_{1}(\textbf{x}_{1})\star_{K}\cdots\star_{K} f_{m}(\textbf{x}_{m}))_{\otimes}  \\
        =\exp\{\hbar\sum\limits_{i<j}\mathcal{K}_{ij}\} 
        ((f_{1}(\textbf{x}_{1})\otimes\cdots\otimes f_{m}(\textbf{x}_{m}))
\end{array}
\end{equation}
where $\mathcal{K}_{ij}$ are bi-vector fields 
acting on $f_{i}(\textbf{x}_{i})$
and $f_{j}(\textbf{x}_{j})$, $i<j$, that is

$$
\begin{array}{c}
\mathcal{K}_{ij}(f_{1}(\textbf{x}_{1})\otimes\cdots\otimes f_{m}(\textbf{x}_{m}))=\\
K_{i,j}(f_{1}(\textbf{x}_{1})\otimes\cdots\otimes
\partial_{i}f_{i}(\textbf{x}_{i})\otimes\cdots 
\otimes\partial_{j}f_{j}(\textbf{x}_{j})\otimes\cdots
\otimes f_{m}(\textbf{x}_{m})).
\end{array}
$$
In this sense $\mathcal{K}_{ij}$ denotes 
a poly-differential operator.

Recalling the discussions about the
generating function of the Bernoulli graphs,
we have

\begin{equation}
\exp\{\hbar\sum\limits_{i<j}\mathcal{K}_{ij}\}
=\sum\limits_{M\in M_{adj}(m,\mathbb{N})}
\frac{\hbar^{degM}}{M!}\mathcal{K}_{M},
\end{equation}
Here $\mathcal{K}_{M}$
are poly-differential operators with coefficients in $\mathcal{A}$.
But the form of $\mathcal{K}_{M}=\prod_{i<j}\mathcal{K}_{ij}^{m_{ij}}$ 
and form of the Feynman amplitudes
are very much alike. It is seems that here $\mathcal{K}_{ij}$ play the role of the 
propagators. In fact for the special choice of $f_{i}(\textbf{x})$, $\mathcal{K}_{ij}$
contributes the propagator indeed, we will discuss that in section 4 and section 5 furthermore.
Therefore we call $\prod_{i<j}\mathcal{K}_{ij}^{m_{ij}}$ the generalised Feynman amplitude.
It is worth to point-out that the generalised Feynman amplitudes appearing in the coefficients
of the star product.
Kontsevich's rule as mentioned above suggests that
the Kontsevich graphs should result in the Feynman graphs. 

Recalling discussion about the connection between the adjacency matrices and
the Feynman amplitudes, it is obvious that there is a way to build the
connection between the Bernoulli graphs and Feynman graphs. 
Now we get one of our main consequence immediately.

\begin{theorem}
There is a one-one correspondence between the Bernoulli graphs 
and Feynman graphs without loops.
\end{theorem}

\begin{proof}
With the help of Kontsevich's rule

$$
\mathcal{U}(\cdot,\Gamma_{M})=
\mathcal{K}_{M},\,M\in M_{adj}(m,\mathbb{N}),\,degM=k,
$$
for a given Bernoulli graph 
$\Gamma_{M}=\prod_{i<j}b_{ij}^{m_{ij}}\in B_{k,m}$, 
it reduces a Feynman graph by the following rule:
\begin{itemize} 
  \item every vertex of the second type is assigned to a vertex of Feynman graph,
  \item every $b_{ij}$ is assigned to a edge of Feynman graph connecting $i$-th  
  and $j$-th vertices of Feynman graph.
\end{itemize}
Consequently, according to above rule, 
from the graph $\prod_{i<j}b_{ij}^{m_{ij}}\in B_{k,m}$
we get a Feynman graph with $m$ vertices, $k$ edges, here there are $m_{ij}$
edges between the $i$-th and $j$-th vertices, $\sum_{i<j}m_{ij}=k$.

Conversely, from a given Feynman graph we can get a graph of Bernoulli type
in an obvious way.
\end{proof}

The graph $b_{M}$ have same form as Feynman amplitudes. 
Due to the correspondence between the Bernoulli graph $b_{ij}$ and bi-vector field
$\mathcal{K}_{ij}$, the graph $\prod_{i<j}b_{ij}^{m_{ij}}$ can be viewed as graphical
version of the Feynman amplitudes.

\section{Wick theorem and Wick power}

In this section we want to discuss the Wick theorem and Wick power 
in terms of the star product at level of functions in $C^{\infty}(\mathbb{R})$.
Therefore we focus on the situation of star product with special form
as $f_{1}(x_{1})\star_{K}\cdots\star_{K} f_{d}(x_{d})$, where 
$f_{i}(\cdot)\in C^{\infty}(\mathbb{R}),\,i=1,\cdots,d$.
Recalling proposition 2.2 we know that

\begin{equation}
    f_{1}(x_{1})\star_{K}\cdots\star_{K} f_{d}(x_{d})
    =\exp\{\hbar\sum\limits_{i<j}K_{ij}\partial_{i}\partial_{j}\}
    f_{1}(x_{1})\cdots f_{d}(x_{d}).
\end{equation}

By the notation $\mathcal{K}_{ij}
=K_{ij}\partial_{i}\otimes\partial_{j}$
and discussions concerning the generating
function of Bernoulli graphs we have

$$
\exp\{\hbar\sum\limits_{i<j}\mathcal{K}_{ij}\}=
\sum\limits_{M\in M_{adj}(d,\mathbb{N})}
\frac{\hbar^{degM}}{M!}\mathcal{K}_{M}.
$$
More precisely we have 

\begin{proposition}
\begin{equation}
\begin{array}{c}
f_{1}(x_{1})\star_{K}\cdots\star_{K} f_{d}(x_{d})   \\
=\sum\limits_{M\in M_{adj}(d,\mathbb{N})}
\frac{\hbar^{degM}}{M!}K_{M}\partial^{\alpha} 
(f_{1}(x_{1})\cdots f_{d}(x_{d})),
\end{array}
\end{equation}
where $\alpha_{i}=\sum_{j}m_{ij}$ ($i=1,\cdots,d$),
and $K_{M}=\prod_{i<j}K_{ij}^{m_{ij}}$.
\end{proposition}

In the formula (4.2) the factors $\prod\limits_{i<j}K_{ij}^{m_{ij}}$ reduced from 
$\prod\limits_{i<j}\mathcal{K}_{ij}^{m_{ij}}$ are very close to
the original Feynman amplitudes in the standard quantum fields theory.

Now we turn to the discussion of Wick power. In this case, it is necessary
for us to consider the following star product:
$$
   \begin{array}{c}
      \underbrace{f(x)\star_{K}\cdots\star_{K} f(x)} ,   \\
       l-times  
\end{array}
$$
where $f(\cdot)\in C^{\infty}(\mathbb{R})$ and $i=1,\cdots,d$.
By definition 2.3 we know

$$
f(x)\star_{K}\cdots\star_{K} f(x)
=(f(x_{1})\star_{K}\cdots\star_{K} f(x_{l}))|_{x_{1}=\cdots=x_{l}=x}.
$$
Similar to the formula (4.2) we have:
\begin{proposition}
\begin{equation}
\begin{array}{c}
     \begin{array}{c}
      \underbrace{f(x)\star_{K}\cdots\star_{K} f(x)} ,   \\
       l-times  
\end{array}       \\
       =\sum\limits_{k\geqslant 0}\frac{\hbar^{k}}{k!}
         \sum\limits_{degM=2k}
    \begin{pmatrix}
       k   \\
       m_{12} ,\cdots,m_{l-1,l}
\end{pmatrix}
K_{M}f^{(\alpha_{1})}(x)\cdots f^{(\alpha_{l})}(x).
\end{array}
\end{equation}
\end{proposition} 

Particularly, if we take $K_{ij}=\Lambda$ for
all $1\leq i<j\leq m$, where $\Lambda\in\mathcal{A}$,
then we have

$$
\begin{array}{c}
\begin{array}{c}
\underbrace{f(x)\star_{K}\cdots\star_{K} f(x)} ,   \\
l-times  
\end{array}       \\
=\sum\limits_{k\geqslant 0}\frac{(\hbar\Lambda)^{k}}{k!}
\sum\limits_{degM=k}
\begin{pmatrix}
k   \\
m_{12} ,\cdots,m_{l-1,l}
\end{pmatrix}
f^{(\alpha_{1})}(x)\cdots f^{(\alpha_{l})}(x).
\end{array}
$$

Furthermore, if we take $f(x_{i})=x_{i}$ we have

\begin{corollary}
\begin{equation}
\begin{array}{cc}
      \underbrace{x_{i}\star_{K}\cdots\star_{K} x_{i}} &   
      =\sum\limits_{k=0}^{[\frac{l}{2}]}\frac{l!}{2^{k}k!(l-2k)!}\,\hbar^{k}\,\Lambda^{k}\,x_{i}^{l-2k}.\\
       l-times &   
\end{array}
\end{equation}
\end{corollary} 

\begin{proof}
It is enough for us to count the number of terms with form $\hbar^{k}\,K_{ii}^{k}\,x_{i}^{l-2k}$.
We consider the star product
$x_{1}\star_{K}\cdots\star_{K}x_{l}$. 
From the formula (4.2) we have

$$
x_{1}\star_{K}\cdots\star_{K}x_{l}
=\sum\limits_{k=0}^{\left[\frac{l}{2}\right]}
\frac{(\hbar\Lambda)^{k}}{k!}
(\sum\limits_{i<j}\partial_{i}\partial_{j})^{k}
(x_{1}\cdots x_{l}).
$$
It is obvious that

$$
(\sum\limits_{i<j}\partial_{i}\partial_{j})^{k}
(x_{1}\cdots x_{l})=
\sum\limits_{i_{1},j_{1},\cdots,i_{k},j_{k}}
\partial_{i_{1}}\partial_{j_{1}}\cdots
\partial_{i_{k}}\partial_{j_{k}}(x_{1}\cdots x_{l}),
$$
where $i_{\alpha}<j_{\alpha}$ ($\alpha=1,\cdots,k$),
$\{i_{\alpha},j_{\alpha}\}\cap\{i_{\beta},j_{\beta}\}
=\emptyset$ ($\alpha\not=\beta$).
The number of the terms in above sum should be

$$
\left(
\begin{array}{c}
l \\ 2
\end{array}
\right)
\left(
\begin{array}{c}
l-2 \\ 2
\end{array}
\right)\cdots
\left(
\begin{array}{c}
l-2k+2 \\ 2
\end{array}
\right)
=\frac{l!}{2^{k}(l-2k)!}.
$$
Finally, let $x_{1}=\cdots=x_{l}=x$,
we get the formula (4.4).

\end{proof}

\begin{definition}
The Wick power in the sense of the star product $\star_{K}$ of $x_{i}$
is defined to be 
\begin{equation}
    \begin{array}{cc}
      :x_{i}^{l}:_{K}= & \underbrace{x_{i}\star_{K}\cdots\star_{K} x_{i}} ,    \\
      & l-times  
\end{array}
\end{equation}
where $1\leqslant i \leqslant d$.
\end{definition}

\begin{remark}
By definition as above we know that the Wick power belongs to $C^{\infty}_{\mathcal{A}}$.
Duo to the formula (4.4) the Wick power is expressed by means of Hermite polynomials.
\end{remark}

Similar to the theorem 2.2, about the case of $f_{1}(x_{1})\star_{K}\cdots\star_{K} f_{d}(x_{d})$
we have more precise formula which is a generalisation of the classical Wick theorem.

\begin{theorem}(\textbf{Wick theorem})
For different propagator matrixes $K=(K_{ij})_{d\times d}$ and $K^{\prime}=(K^{\prime}_{ij})_{d\times d}$
we have

\begin{equation}
\begin{array}{c}
f_{1}(x_{1})\star_{K}\cdots\star_{K} f_{d}(x_{d})   \\
=\sum\limits_{M\in M_{adj}(d,\mathbb{N})}
\frac{\hbar^{degM}}{M!}(K-K^{\prime})_{M}
f_{1}^{(\alpha_{1})}\star_{K^{\prime}}\cdots\star_{K^{\prime}} f_{d}^{(\alpha_{d})}.   
\end{array}
\end{equation}
Where $(K-K^{\prime})_{M}=\sum_{i<j}(K_{ij}-K_{ij}^{\prime})^{m_{ij}}$ and

$$
\alpha_{i}=\sum\limits_{j}m_{ij},\,i=1,\cdots,d.
$$
\end{theorem}
\begin{proof}
Similar to the proof of theorem 2.2 we have
$$
    \exp\{\hbar\sum\limits_{i<j}K_{ij}\partial_{i}\partial_{j}\}
    = \exp\{\hbar\sum\limits_{i<j}(K_{ij}-K^{\prime}_{ij})\partial_{i}\partial_{j}\}
     \exp\{\hbar\sum\limits_{i<j}K^{\prime}_{ij}\partial_{i}\partial_{j}\}.$$
Therefore
$$
    f_{1}(x_{1})\star_{K}\cdots\star_{K} f_{d}(x_{d})
    = \exp\{\hbar\sum\limits_{i<j}(K_{ij}-K^{\prime}_{ij})\partial_{i}\partial_{j}\}
         f_{1}(x_{1})\star_{K^{\prime}}\cdots\star_{K^{\prime}} f_{d}(x_{d}).
$$
Above formula implies (4.6). 
\end{proof}

\begin{remark}
\begin{itemize}
$\\$
  \item If we take $K^{\prime}=0$ in the formula (4.6), we come bake to the formula
(4.2).
  \item If we take $K=0$, we have
$$
\begin{array}{c}
f_{1}(x_{1})\cdots f_{d}(x_{d})   \\
=\sum\limits_{M\in M_{adj}(d,\mathbb{N})}
\frac{(-\hbar)^{degM}}{M!}K^{\prime}_{M}
f_{1}^{(\alpha_{1})}\star_{K^{\prime}}\cdots\star_{K^{\prime}} f_{d}^{(\alpha_{d})}.        
\end{array}
$$
Specially, we can get inversion of the formula (4.4):
$$
x^{l}=\sum\limits_{k=0}^{[\frac{l}{2}]}\frac{l!}{2^{k}k!(l-2k)!}\,(-\hbar)^{k}\,\Lambda^{k}\,:x^{l-2k}:_{K}.
$$
\end{itemize}
\end{remark}

If we make a special choice of $f_{i}(x_{i})$ in the formula (4.6), i.e.
 $:x_{i}^{n_{i}}:_{K}\in C^{\infty}_{\mathcal{A}}$,
$n_{i}\in\mathbb{N},\,i=1,\cdots,d$, we can get the Wick theorem with expression very
similar to the classical Wick theorem in the standard quantum fields theory. 
\begin{corollary}
\begin{equation}
    \begin{array}{c}
       :x_{1}^{n_{1}}:_{K}\star_{K^{(1)}}\cdots\star_{K^{(1)}} :x_{d}^{n_{d}}:_{K}   \\
       =\sum\limits_{M\in M_{adj}(d,\mathbb{N})}\frac{\hbar^{degM}}{M!} 
(K^{(1)}-K^{(2)})_{M}       
\begin{pmatrix}
       n_{1}   \\
       \alpha_{1} 
\end{pmatrix}
\cdots
\begin{pmatrix}
      n_{d}    \\
      \alpha_{d}  
\end{pmatrix}\cdot \\
:x_{1}^{n_{1}-\alpha_{1}}:_{K}\star_{K^{(2)}}\cdots\star_{K^{(2)}}:x_{d}^{n_{d}-\alpha_{d}}:_{K}.
\end{array}
\end{equation}
Where 

$$
\alpha_{i}=\sum\limits_{j}m_{ij},\,\alpha_{i}\leqslant n_{i}\,i=1,\cdots,d.
$$
\end{corollary}

\begin{remark}
  Observing the formula (4.7), we find that the form of Feynman amplitudes
  dose not dependent on the choices of propagator matrixes $K^{(1)}$ and $K^{(2)}$.
  Thus, when we focus on the issues of the Feynman amplitudes,
  without loss of generality we can choose the the star product $\star_{K^{(1)}}$
  in Wick-monomial  
  $:x_{1}^{n_{1}}:_{K}\star_{K^{(1)}}\cdots\star_{K^{(1)}} :x_{d}^{n_{d}}:_{K}$
  to be commutative always.
\end{remark}

In the traditional sense the Feynman amplitudes arise from the expectation of 
Green functions. in the present paper Feynman amplitudes arise from the
bi-vector fields or coefficeints of the star product.
Actually, in perturbative algebraic quantum fields
theory the expectation of the Wick-monomial can be defined as the coefficient
of the term with the highest power of $\hbar$ in the star product. Now we define
the expectation of Wick-monomial $:x_{1}^{n_{1}}:_{K}\star_{K^{(1)}}\cdots\star_{K^{(1)}} :x_{d}^{n_{d}}:_{K}$,
denoted by

$$
<:x_{1}^{n_{1}}:_{K}\star_{K^{(1)}}\cdots\star_{K^{(1)}} :x_{d}^{n_{d}}:_{K}>,
$$
as following:

\begin{definition}
If we write the Wick-monomial as a polynomial of $\hbar$

$$
:x_{1}^{n_{1}}:_{K}\star_{K^{(1)}}\cdots\star_{K^{(1)}} :x_{d}^{n_{d}}:_{K}=\sum\limits_{k=1}^{m}c_{k}\hbar^{k},
$$
we define the expectation of above Wick-monomial as following:

\begin{itemize}
  \item When $n_{1}+\cdots+n_{d}=2m$,
  
  \begin{equation}
  <:x_{1}^{n_{1}}:_{K}\star_{K^{(1)}}\cdots\star_{K^{(1)}} :x_{d}^{n_{d}}:_{K}>=c_{m},
  \end{equation}
where

\begin{equation}
c_{m}=
\sum\limits_{M\in M_{adj}(d,\mathbb{N}),\,degM=m}
\frac{K_{M}^{(1)}}{M!}, 
\end{equation}
above sum is over all adjacency matrices $M=(m_{ij})_{d\times d}$, $degM=2m$, such that
$$n_{i}=\sum\limits_{j}m_{ij},\,i=1,\cdots,d.$$
  \item When $n_{1}+\cdots+n_{d}>2m$, 
  $$<:x_{1}^{n_{1}}:_{K}\star_{K^{(1)}}\cdots\star_{K^{(1)}} :x_{d}^{n_{d}}:_{K}>=0.$$
  \end{itemize}
\end{definition}

\begin{definition}
$\\$
When the integer sequence $(n_{1},\cdots,n_{d})$ satisfies the following conditions:
\begin{itemize}
  \item There is an adjacency matrix $M=(m_{ij})_{d\times d}$, $degM=2m$, such that
  \begin{equation}
    \begin{array}{c}
      2m=n_{1}+\cdots+n_{d},    \\
      n_{i}=\sum\limits_{j}m_{ij},\,i=1,\cdots,d,  
    \end{array}
 \end{equation}
\end{itemize}
we call this integer sequence $(n_{1},\cdots,n_{d})$ admissible.
\end{definition}

Combining the definition 4.2 and 4.3 we have the following conclusion immediately.
\begin{proposition}
The Wick-monomial $:x_{1}^{n_{1}}:_{K}\star_{K^{(1)}}\cdots\star_{K^{(1)}} :x_{d}^{n_{d}}:_{K}$
endows non-zero expectation iff $(n_{1},\cdots,n_{d})$ is admissible.
\end{proposition}

We have a simpler description of the admissible integer sequence.
Here we assume the star product is commutative and $n_{i}>0,\,i=1,\cdots,d$.
\begin{theorem}
A integer sequence $(n_{1},\cdots,n_{d})$ is admissible iff $n_{1}+\cdots+n_{d}=2m$ 
and $max_{1\leqslant i\leqslant d}\{n_{i}\}\leqslant m$,
where $m$ is a positive integer.
\end{theorem}
\begin{proof}
If the integer sequence
$(n_{1},\cdots,n_{d})$ is admissible, i.e. there is an adjacency matrix $M=(m_{ij})_{d\times d}$,
such that $n_{i}=\sum\limits_{j}m_{ij},\,i=1,\cdots,d$. Then
we have
$$2n_{i}=\sum\limits_{j}m_{ij}+\sum\limits_{j}m_{ji}\leqslant 
\sum\limits_{i,j}m_{ij}=2\sum\limits_{i<j}m_{ij}=\sum\limits_{i}n_{i}.$$

Conversely we need to prove that for an integer sequence $(n_{1},\cdots,n_{d})$
satisfying $n_{1}+\cdots+n_{d}=2m,\,m\in \mathbb{N}$ and $2n_{i}\leqslant n_{1}+\cdots+n_{d},\,i=1,\cdots,d$
there is an adjacency matrix $M=(m_{ij})_{d\times d}$
such that $n_{i}=\sum\limits_{j}m_{ij},\,i=1,\cdots,d$.
Now we prove  the existence of adjacency matrix by induction for $d$.
Without loss of generality we assume $n_{1}\geqslant\cdots\geqslant n_{d}$.

When $d=2$, then $n_{1}=n_{2}$ at this time. In this case there is an unique suitable adjacency
matrix
$$M=
\begin{pmatrix}
     0 & n_{1}   \\
     n_{2} &  0
\end{pmatrix}$$
satisfying the conditions what we need.

Suppose the conclusion is valid for $d$, now we consider the case of $d+1$.
The case will be divided into a few parts.

\textbf{Case of $n_{1}=\cdots=n_{d+1}$:}
\begin{itemize}
  \item $d+1$ is an even integer: Let $n_{1}=\cdots=n_{d+1}=p$,
  we can take the entries of $M=(m_{ij})_{(d+1)\times(d+1)}$ to be
  $$m_{ij}=p,i+j=d+2;m_{ij}=0,\,for\,others.$$
  \item $d+1$ is an odd integer: Because $(d+1)p=2m$, p is an even integer,
  let $p=2q$, the entries of $M=(m_{ij})_{(d+1)\times(d+1)}$ can be taken
  to be 
  $$m_{i,i+1}=m_{i+1,i}=q,i=1,\cdots,d,m_{1,d+1}=m_{d+1,1}=q;m_{ij}=0,\,for\,others.$$
\end{itemize}

\textbf{Case of $n_{1}>n_{d+1}$:}
$\\$

Let $n^{\prime}_{1}=n_{1}-n_{d+1}$, then we know that 
$$n^{\prime}_{1}+n_{2}+\cdots+n_{d}=\sum\limits_{i=1}^{d+1}n_{i}-2n_{d+1}$$
is an even integer and $(n^{\prime}_{1},n_{2},\cdots,n_{d})$
satisfies the condition (4.11).
According to the hypothesis of induction we know that there is an adjacency matrix 
$M=(m_{ij})_{d\times d}$ such that
$$n^{\prime}_{1}=\sum\limits_{j}m_{1j},n_{i}=\sum\limits_{j}m_{ij},i=2,\cdots,d.$$

Now we take a $(d+1)\times(d+1)$ adjacency matrix as following:
$$
M_{1}=
\begin{pmatrix}
      M &
\begin{array}{c}
      n_{d+1}    \\
        0          \\
        \vdots
\end{array}    \\
       n_{d+1}\,\,0\cdots   & 0 
\end{pmatrix}.
$$
The matrix $M_{1}$ satisfies all conditions what we need.
\end{proof}

We want to talk about theorem 4.2 more from the combinatorial viewpoint.
Under the the assumption being $n_{i}>0,\,i=1,\cdots,d$ the adjacency
matrixes need additional restriction which is that there is at least one
positive entry at every row or column. Recalling RSK algorithm, there are 
one-one correspondences between the following objects:
(see \cite{3}, \cite{4}, \cite{14})
\begin{itemize}
  \item the set of permutations which are involutions without fixed points,
  \item the set of adjacency matrixes with zeros along the main diagonal,
  \item the set of semi-standard Young tableaus(SSYT) without odd columns.
\end{itemize}

From the combinatorial viewpoint the integer sequence $(n_{1},\cdots,n_{d})$
arising from the monomial $:x_{1}^{n_{1}}:_{K}\star_{K^{(1)}}\cdots\star_{K^{(1)}} :x_{d}^{n_{d}}:_{K} $
plays the role of content for some semi-standard Young tableau denoted 
by $1^{n_{1}}2^{n_{2}}\cdots d^{n_{d}}$ usually. Here we assume 
$n_{1}\geqslant\cdots\geqslant n_{d}$. Starting from an admissible integer sequence
$(n_{1},\cdots,n_{d})$, now we begin to construct
a special SSYT such that under the one-one correspondence mentioned above
this Young tableau results in an adjacency matrix which satisfies
the conditions in definition 4.2. Firstly we construct an Young diagram
with 2 rows and there are $m$ columns in each row, where
$m=\frac{1}{2}\sum_{i}n_{i}$. Secondly we fill the numbers in above Young
diagram as follows. At beginning we put the numbers in the first row.
Starting from the up-left corner of Young diagram we put number
"1" with $n_{1}$ times, and then we put all of "2" and so on until the first
row is full. Continuously we put numbers in the second row in the same way.
Consequently we get a SSYT as table 3, where 
$2\leqslant r_{1}\leqslant r_{2}\leqslant r_{3}\leqslant r_{4}\leqslant d$.
Above discussion gives a combinatorial proof of theorem 4.2.
\begin{table}
\begin{center}
\begin{tabular}{|c|c|c|c|c|c|}
\hline 1 & $\cdots$ & 1 & 2 & $\cdots$ & $r_{1}$ \\
\hline $r_{2}$ & $\cdots$ & $r_{3}$ & $r_{4}$ &$ \cdots$ & d \\
\hline 
\end{tabular}
 \caption{default}
\end{center}
\label{defaulttable}
\end{table}

\section{The star products at levels of fields and functionals}

In this section we construct the star products at levels of fields and functionals
base on the star product of functions discussed in previous sections.
We will discuss the problems on the general smooth manifold which contains
Lorentzian manifold as a special case.

\paragraph{The star product of the scalar fields}

Let $X$ be a $n$-dimensional real smooth manifold, $K(\textbf{x},\textbf{y})\in \mathcal{D}^{\prime}(X\times X)$. 
For simplicity we assume $K(\textbf{x},\textbf{y})\in C^{\infty}(X\times X)$ which can be considered as 
regulation of general distribution on $X\times X$. In this section the bold letters $\textbf{x}$ or $\textbf{y}$
will denote the points in $X$.
At beginning of this subsection we discuss the star product of fields.
Here we discuss the star product of functions being of form

$$
\begin{array}{cc}
      f(\varphi(\textbf{x}_{1}),\cdots,\varphi(\textbf{x}_{d}))\in &C^{\infty}(\underbrace{X\times\cdots\times X}),    \\
      & d-times  
\end{array} 
$$
where $f(y_{1},\cdots,y_{d})\in C^{\infty}(\mathbb{R}^{d})$, the function $\varphi(\textbf{x})\in C^{\infty}(X)$
plays the role of real scalar field on $X$.

\begin{definition}
Let $K(\textbf{x},\textbf{y})\in C^{\infty}(X\times X)$, $f(y_{1},\cdots,y_{d})$,
$g(y_{1},\cdots,y_{d})\in C^{\infty}(\mathbb{R}^{d})$, $\varphi(\textbf{x})\in C^{\infty}(X)$, we define
the star product as following:

\begin{equation}
\begin{array}{c}
       f(\varphi(\textbf{x}_{1}),\cdots,\varphi(\textbf{x}_{d}))\star_{K}
        g(\varphi(\textbf{y}_{1}),\cdots,\varphi(\textbf{y}_{d}))  \\
        =\textbf{m}\circ \left[\exp\{\hbar\mathcal{K}\}
        f(x_{1},\cdots,x_{d})\otimes g(y_{1},\cdots,y_{d}) |_{x_{i}=\varphi(\textbf{x}_{i}),y_{i}=\varphi(\textbf{y}_{i})}\right].
\end{array}
\end{equation}
Where $\mathcal{K}=\sum_{i,j}K_{ij}\partial_{x_{i}}\otimes\partial_{y_{j}}$, $K_{ij}=K(\textbf{x}_{i},\textbf{y}_{j})$.  
\end{definition}

When $K(\textbf{x},\textbf{y})=P^{\ast}K(\textbf{x},\textbf{y})$, where 
$P:X\times X\longrightarrow X\times X; P(\textbf{x},\textbf{y})=(\textbf{y},\textbf{x})$, is permutation
map, we know that the star product $\star_{K}$ is commutative. 
For non-commutative case we have:

\begin{definition}
The Poisson bracket is defined to be

\begin{equation}
\begin{array}{c}
       \{f(\varphi(\textbf{x}_{1}),\cdots,\varphi(\textbf{x}_{d})),g(\varphi(\textbf{y}_{1}),\cdots,\varphi(\textbf{y}_{d}))\}_{K}   \\
      =\textbf{m}\circ\left[\sum_{i,j}(K_{ij}-K_{ji})\partial_{i}f(x_{1},\cdots,x_{d})\otimes\partial_{j}g(y_{1},\cdots,y_{d})
      | _{x_{i}=\varphi(\textbf{x}_{i}),y_{i}=\varphi(\textbf{y}_{i})}\right] .
\end{array}
\end{equation}
\end{definition}

\begin{remark}
$\\$
\begin{itemize}
  \item The star product defined in definition 5.1 and Poisson bracket defined in definition 5.2 
  are well defined and rely on the issues at level functions. Actually as a special case we have
  
\begin{equation}
\begin{array}{c}
        f(\varphi(\textbf{x}_{1}),\cdots,\varphi(\textbf{x}_{d}))\star_{K}
        g(\varphi(\textbf{x}_{1}),\cdots,\varphi(\textbf{x}_{d}))    \\
        = f(\varphi(\textbf{x}_{1}),\cdots,\varphi(\textbf{x}_{d}))\star_{K}
        g(\varphi(\textbf{y}_{1}),\cdots,\varphi(\textbf{y}_{d})) 
        |_{\textbf{x}_{i}=\textbf{y}_{i}} \\
        =f(x_{1},\cdots, x_{d})\star_{K} g(y_{1},\cdots,y_{d})
        |_{x_{i}=y_{i}=\varphi(\textbf{x}_{i})}.      
        \end{array}
\end{equation}
and

\begin{equation}
\begin{array}{c}
       \{f(\varphi(\textbf{x}_{1}),\cdots,\varphi(\textbf{x}_{d})),g(\varphi(\textbf{x}_{1}),\cdots,\varphi(\textbf{x}_{d}))\}_{K}    \\
       =\{f(y_{1},\cdots,y_{d}),g(y_{1},\cdots,y_{d}) \} |_{y_{i}=\varphi(\textbf{x}_{i})} .
\end{array}
\end{equation}  
Where $K_{ij}=K(\textbf{x}_{i},\textbf{x}_{j})$ .
  \item For all of conclusions in section 2 there are parallel ones in the case of scalar fields. 
\end{itemize}
\end{remark}

In the case of the scalar fields in the 
sense of our setting we have 
Wick theorem and Wick power similar to
the situation of the standard quantum fields theory. 
Recalling the discussion in section 4, 
if we take $x_{i}=\varphi(\textbf{x}_{i}),\,i=1,\cdots,d$, in each formula in section 4,
we can get a corresponding formula in the case of fields.

\begin{theorem}
Let $K(\textbf{x},\textbf{y}),\,K^{\prime}(\textbf{x},\textbf{y})$ be smooth functions
on $X\times X$, $f_{1}(\cdot),\cdots,f_{d}(\cdot)\in C^{\infty}(\mathbb{R})$, $\varphi(\textbf{x})\in C^{\infty}(X)$,
we have
\begin{equation}
\begin{array}{c}
        f_{1}(\varphi(\textbf{x}_{1}))\star_{K}\cdots\star_{K}f_{d}(\varphi(\textbf{x}_{d})) \\
        =\sum\limits_{M\in M_{adj}(d,\mathbb{N})}\frac{\hbar^{degM}}{M!}
     (K-K^{\prime})_{M}\,  
    f_{1}^{(\alpha_{1})}(\varphi(\textbf{x}_{1}))\star_{K^{\prime}}\cdots\star_{K^{\prime}} f_{d}^{(\alpha_{d})}(\varphi(\textbf{x}_{d})), 
\end{array}
\end{equation}
where

$$
\begin{array}{c}
      \alpha_{i}=\sum\limits_{j}m_{ij},\,i=1,\cdots,d,   \\
      K_{ij}=K(\textbf{x}_{i},\textbf{x}_{j}),\,  K^{\prime}_{ij}=K^{\prime}(\textbf{x}_{i},\textbf{x}_{j}) .
\end{array}
$$
\end{theorem}

Moreover we define the Wick power in the case of fields as following:
\begin{equation}
    \begin{array}{cc}
    :\varphi^{l}(\textbf{x}):_{K}= & \underbrace{\varphi(\textbf{x})\star_{K}\cdots\star_{K} \varphi(\textbf{x})} .\\
          &   l-times
    \end{array}
\end{equation}

Precisely we have an expression of Wick power in terms of Hermite polynomials

\begin{equation}
    :\varphi^{l}(\textbf{x}_{i}):_{K}=\sum\limits_{k=0} ^{[\frac{l}{2}]}\frac{l!}{2^{k}(l-2k)!k!}
    \,\hbar^{k}\,K_{ii}^{k}\,(\varphi(\textbf{x}_{i}))^{l-2k},\,1\leqslant i\leqslant d,
\end{equation}
where $K_{ii}=K(\textbf{x}_{i},\textbf{x}_{i}),\,1\leqslant i\leqslant d$.

The following Wick theorem expressed by means of Wick power is more closed to
classical Wick theorem.

\begin{corollary}
Let $K(\textbf{x},\textbf{y}),K^{(1)}(\textbf{x},\textbf{y}),K^{(2)}(\textbf{x},\textbf{y})
\in C^{\infty}(X\times X)$, then

\begin{equation}
\begin{array}{c}
:\varphi^{n_{1}}(\textbf{x}_{1}):_{K}\star_{K^{(1)}}
\cdots\star_{K^{(1)}}:\varphi^{^{n_{d}}}(\textbf{x}_{d}):_{K}    \\
=\sum\limits_{M\in M_{adj}(d,\mathbb{N})}
\frac{\hbar^{degM}}{M!}(K^{(1)}-K^{(2)})_{M}
\begin{pmatrix}
      n_{1}    \\
      \alpha_{1}  
\end{pmatrix}\cdots
\begin{pmatrix}
      n_{d}    \\
      \alpha_{d}  
\end{pmatrix}\cdot \\
:\varphi^{n_{1}-\alpha_{1}}(\textbf{x}_{1}):_{K}
\star_{K^{(2)}}\cdots\star_{K^{(2)}}
:\varphi^{^{n_{d}-\alpha_{d}}}(\textbf{x}_{d}):_{K},
\end{array}
\end{equation}
where

$$
\begin{array}{c}
      \alpha_{i}=\sum\limits_{j}m_{ij},\,\alpha_{i}\leqslant n_{i},\,\,i=1,\cdots,d,   \\
      K^{(1)}_{ij}=K^{(1)}(\textbf{x}_{i},\textbf{x}_{j}),\,  K^{(2)}_{ij}=K^{(2)}(\textbf{x}_{i},\textbf{x}_{j}) .
\end{array}
$$
\end{corollary}

Now we define the expectation of monomial 
$:\varphi^{n_{1}}(\textbf{x}_{1}):_{K}\star_{K^{\prime}}\cdots\star_{K^{\prime}}:\varphi^{^{n_{d}}}(\textbf{x}_{d}):_{K}$,
where $K(\textbf{x},\textbf{y}),K^{\prime}(\textbf{x},\textbf{y})$ are smooth functions on $X\times X$.
For convenience we assume the star product $\star_{K^{\prime}}$ is commutative.
\begin{definition}
$\\$
\begin{itemize}
  \item When the integer sequence $(n_{1},\cdots,n_{d})$ is admissible, we define
  
  \begin{equation}
\begin{array}{c}
       <:\varphi^{n_{1}}(\textbf{x}_{1}):_{K}\star_{K^{\prime}}\cdots\star_{K^{\prime}}:\varphi^{^{n_{d}}}(\textbf{x}_{d}):_{K}>   \\
      =\sum\limits_{M\in M_{adj}(d,\mathbb{N}),degM=m} 
\frac{K^{\prime}_{M}}{M!},
\end{array}
\end{equation}
where $2m=\sum_{i}n_{i}$,  the sum in (5.9) is over all adjacency matrixes $M=(m_{ij})_{d\times d}$
with $degM=2m$ satisfying
$$n_{i}=\sum_{j}m_{ij},\,i=1,\cdots,d.$$
  \item for others
  $$<:\varphi^{n_{1}}(\textbf{x}_{1}):_{K}\star_{K^{\prime}}\cdots\star_{K^{\prime}}:\varphi^{^{n_{d}}}(\textbf{x}_{d}):_{K}>=0.$$
\end{itemize}
\end{definition}

\begin{remark}
For general distribution $K(\textbf{x},\textbf{y})\in \mathcal{D}^{\prime}(X\times X)$,
the power and restriction on diagonal of $X\times X$ 
make non-sense generally. In this case only the star product with form
$\varphi(\textbf{x}_{1})\star_{K}\cdots\star_{K}\varphi(\textbf{x}_{d})$ may be well defined,
but some analytic conditions, for example, concerning wave front set $WF(K)$, may be
needed.
\end{remark}

\paragraph{The star product of the functionals}

Now we turn to situation of the functionals. We consider the functionals with form
\begin{equation}
    F(\varphi)= \int_{X^{d}}
      f(\textbf{x}_{1},\cdots,\textbf{x}_{d},\varphi(\textbf{x}_{1}),\cdots,\varphi(\textbf{x}_{d}))
      dV_{d},    
\end{equation}
where $f\in C^{\infty}(X^{d}\times \mathbb{R}^{d})$, 
$$
\begin{array}{cc}
      X^{d}= &\underbrace{X\times\cdots\times X} ,   \\
      & d-times   
\end{array}
$$
and $dV_{d}$ is volume form on $X^{d}$.

In the below discussion we assume the integrals make sense always. 
We state the definitions of star product and Poisson bracket of functionals
as following.
\begin{definition}
Let $F(\varphi),G(\varphi)$ be functionals as in (5.10).
\begin{itemize}
  \item We define their star product to be
\begin{equation}
        F(\varphi)\star_{K}G(\varphi) 
      =\int_{X^{d}} \int_{X^{d}} f(\cdot)\star_{K}g(\cdot)dV_{d} dV_{d},
\end{equation}
where
$$
\begin{array}{c}
        f(\textbf{x}_{1},\cdots,\textbf{x}_{d},\varphi(\textbf{x}_{1}),\cdots,\varphi(\textbf{x}_{d}))
        \star_{K}g(\textbf{y}_{1},\cdots,\textbf{y}_{d},\varphi(\textbf{y}_{1}),\cdots,\varphi(\textbf{y}_{d}))   \\
        =\textbf{m} \circ\left[\exp\{\hbar\mathcal{K}\}
       f(\textbf{x}_{1},\cdots,x_{1},\cdots,x_{d})
       g(\textbf{y}_{1},\cdots,y_{1},\cdots,y_{d})|_{x_{i}=\varphi(\textbf{x}_{i}),y_{i}=\varphi(\textbf{y}_{i})} \right].
\end{array}
$$
where $\mathcal{K}=\sum\limits_{i,j}K_{ij}\partial_{x_{i}}\otimes\partial_{y_{j}}$,
$K_{ij}=K(\textbf{x}_{i},\textbf{y}_{j})$.
   \item For the star product between the functionals and fields we have
\begin{equation}
      F(\varphi)\star_{K}g(\textbf{y}_{1},\cdots,\textbf{y}_{d},\varphi(\textbf{y}_{1}),\cdots,\varphi(\textbf{y}_{d}))
      =\int_{X^{d}}f(\cdot)\star_{K}g(\cdot)dV_{d},
\end{equation}
in (5.12) the integral concerns the variables of $f(\cdot)$.
    \item We define the Poisson bracket of $F(\varphi)$ and $G(\varphi)$ to be
\begin{equation}
    \{F(\varphi),G(\varphi)\}_{K}=\int_{X^{d}} \int_{X^{d}}\{ f(\cdot),g(\cdot)\}_{K}dV_{d} dV_{d},
\end{equation} 
where 
$$
\begin{array}{c}
      \{f(\textbf{x}_{1},\cdots,\textbf{x}_{d},\varphi(\textbf{x}_{1}),\cdots,\varphi(\textbf{x}_{d})),
      g(\textbf{y}_{1},\cdots,\textbf{y}_{d},\varphi(\textbf{y}_{1}),\cdots,\varphi(\textbf{y}_{d}))\}_{K}    \\
      =\sum\limits_{i,j}(K_{ij}-K_{ji})\partial_{x_{i}}f(\textbf{x}_{1},\cdots,x_{1},\cdots,x_{d})
      \partial_{y_{j}}g(\textbf{y}_{1},\cdots,y_{1},\cdots,y_{d})|_{x_{i}=\varphi(\textbf{x}_{i}),y_{i}=\varphi(\textbf{y}_{i})} . 
\end{array}
$$
\end{itemize}
\end{definition}

The star product and Poisson bracket defined in definition 5.3 are well
defined and satisfy all conditions which are needed.

\end{document}